\documentclass[11pt]{article}
\usepackage{probsys}

\DeclareMathOperator{\supp}{supp}
\DeclareMathOperator{\dom}{dom}

\DeclarePairedDelimiter{\modZ}{\langle}{\rangle}

\newcommand{\N}{\mathbb{N}}
\newcommand{\R}{\mathbb{R}}
\newcommand{\Q}{\mathbb{Q}}
\newcommand{\Z}{\mathbb{Z}}

\newcommand{\eps}{\varepsilon}

\newcommand{\defeq}{\coloneqq}
\newcommand{\eqdef}{\eqqcolon}

\newcommand{\DAlphabet}{\mathbb{D}}

\newcommand{\Flip}{U}               
\newcommand{\AState}{\mathbb{S}}

\newcommand{\Prob}{\mathbb{P}}

\newcommand{\xinsupp}{\mathclap{x\in \supp(p)}}

\newcommand{\setversion}[1]{
  \ifstrequal{#1}{public}{\includecomment{public}\excludecomment{private}}{
  \ifstrequal{#1}{private}{\excludecomment{public}\includecomment{private}}{
  \PackageError{main}{Unknown version `#1'}{Use `private' or `public'.}}}}

\title{Space--Entropy Lower Bounds for Random Sampling}

\begin{document}

\author{
   Thomas L.~Draper \\
   \href{mailto:tdraper@cmu.edu}{tdraper@cmu.edu}
   \and
   Feras A.~Saad \\
   \href{mailto:fsaad@cmu.edu}{fsaad@cmu.edu}}

\date{Carnegie Mellon University \\[10pt] \today}

\maketitle

\pagenumbering{arabic}
\setcounter{page}{0}
\thispagestyle{empty}

\begin{abstract}
We prove fundamental space lower bounds for exact random sampling using an
entropy source of i.i.d.~uniform bits.
A classic result from information theory shows that generating $n$ discrete random
variables $X_1, \dots, X_n$ requires at least $H(X_1, \dots, X_n)$ input
random bits on average, where $H$ is the \citeauthor{shannon1948} entropy
function.
How much space must a random sampling algorithm use in order
to approach this information-theoretically optimal entropy bound?

We prove that any random sampling algorithm that is exact for arbitrary discrete
target distributions and consumes at most $H(X_1,\ldots,X_n)+\eps n+o(n)$
input bits in expectation for every output process must use
$\Omega(\log(1/\eps))$ bits of space.
In fact, i.i.d.~sampling
from the single distribution $\mathrm{Bernoulli}(1/3)$ already forces
at least $(1/{5.116201}-o(1))\log(1/\eps)$ bits of space.
If the sampler handles a family of infinitely many Bernoulli distributions,
we show a sharper bound of at least $\log(1/\eps)$ bits of space.
We also prove lower bounds for general i.i.d.~sampling: for almost every
distribution on $k$ outcomes, the space is at least
$(1/(k+1)-o(1))\log(1/\eps)$ bits.

The proof technique is based on a graph-theoretic analysis of the amount
of information that any algorithm can store in its state.
Finite state spaces force short cycles around the state-transition graph,
and the loss
around such cycles reduces to Diophantine lower bounds on fractional parts
of integer combinations of log-probabilities.
To the best of our knowledge, these results comprise the first known
space lower bounds for entropy-efficient random sampling.
\end{abstract}

\setcounter{tocdepth}{1}
\tableofcontents

\clearpage

\section{Introduction}
\label{sec:introduction}

We are concerned with the problem of generating a sequence
$X_1, X_2, \dots$ of discrete random variables using an entropy source that emits
independent and identically distributed random bits $\Flip_1, \Flip_2, \dots$.
This problem is known in the literature as \textit{random variate generation}~\citep{devroye1986}
or \textit{random sampling}, and is a central object of study
in theoretical and applied computer science.
In \textit{online} random sampling, the target distribution for each $X_i$
is dynamically presented at the $i$th iteration of sampling, as shown in
the following interaction pattern.

\begin{nicebox}[boxrule=0pt,leftrule=3pt, left skip=\parindent, left=0pt, boxsep=0pt, colframe=black]
\begin{algorithmic}
\State \textbf{Input:} Sequence $\Flip_1, \Flip_2, \dots \overset{\rm iid}{\sim} \mathrm{Uniform}\set{0,1}$
\State For $i=1,2,\dots$
\State \quad Receive next target distribution $p_i$
\State \quad \textbf{Emit} next random variable $X_i \sim p_i$
\end{algorithmic}
\end{nicebox}

Online random sampling arises whenever random variables must be generated
sequentially from distributions that are revealed or updated during
execution, as each $p_i$ may itself depend on the previously generated
variables or auxiliary sources of randomness.
Examples include
autoregressive sequence generation in large language models,
   where the next-token distribution depends on the previously sampled prefix;
particle filtering, where resampling distributions are updated
   after each observation;
bandit and reinforcement-learning algorithms, where
   action distributions depend on the history of actions and rewards; and
discrete-event simulations, where transition probabilities depend on the
current system state.
In all these settings, the random sampler receives target distributions
$p_1, p_2, \dots$ that are presented adaptively over time.

\paragraph{Complexity}
The complexity of an online random sampling algorithm is measured according
to three resources that scale with the iteration $n$: its space usage
$S(n)$, runtime $R(n)$, and entropy cost $T(n)$, which is the average
number of consumed random input bits from the entropy source.
A classic result from \citet{knuth1976} gives an ``entropy-optimal'' sampler
whose entropy cost is tightly concentrated around the \citeauthor{shannon1948} entropy:
\begin{align}
H(X_1, \dots, X_n) \le T(n) < H(X_1, \dots, X_n) + 2,
\end{align}
where $H(Z) \defeq -\sum_{z}\Prob(Z=z) \log(\Prob(Z=z))$ for any discrete random
element $Z$, and all logarithms are taken with base two.
Algorithms such as DDG sampling~\citep{knuth1976} and the
interval algorithm~\citep{han1997} achieve an asymptotically optimal
entropy cost of $T(n) \leq H(X_1, \dots, X_n) + O(1)$, but require unbounded
space $S(n) \to \infty$ as $n \to \infty$.

\paragraph{Main Result on Space--Entropy Lower Bounds}
In recent work, \citet{draper2026ieee} develop a sampling algorithm for
the case that each $p_i$ is a rational distribution with bounded
denominator at most $d$. Its entropy cost satisfies
\begin{align}
T(n) \le H(X_1, \dots, X_n) + \eps n + \log(d/\eps),
\end{align}
using \textit{constant} space $S(n) = O(\log(d/\eps))$.
They also offer a conjecture for the general case.

\begin{nicebox}
\begin{conjecture}[name={\citep{draper2026soda}}]
\label{conj:entropy-cost-tight}
Any online random sampling algorithm that generates exact samples from a sequence of
arbitrary discrete distributions,
within $\varepsilon > 0$ of the information-theoretically
optimal entropy rate,
using a stream of i.i.d.~random bits as the entropy source,
requires $\Omega(\log(1/\varepsilon))$ bits of space
for an auxiliary state that is carried over between rounds.
\end{conjecture}
\end{nicebox}

In this article, we formally prove this conjecture, establishing the first
known space lower bound on entropy-efficient online random sampling.

\begin{nicebox}
\begin{theorem}[store=theorem:space-lb,restate-keys={note=Restated}]
\label{theorem:space-lb}
Any online random sampler that can generate exact samples
from a sequence of arbitrary discrete probability distributions
using an input sequence of i.i.d.~random bits requires $\log(\abs{\AState}) \geq \Omega(\log(1/\eps))$ bits
of persistent space to achieve an entropy cost of $T(n) \leq H(X_1,\dots,X_n) + \eps n + o(n)$,
where $\AState$ is the state space of the algorithm.
\\

In particular, suppose the sampler can generate samples from a family
$\DAlphabet$ of probability distributions over the natural numbers $\N$.
\begin{enumerate}[label={(\zcref*[noname]{theorem:space-lb}.\arabic*)},wide]
\item \label{item:state-space-lower-bounds-1}
If $\mathrm{Bernoulli}(1/3) \in \DAlphabet$,
then $\log(\abs{\AState}) \geq (1/5.116201 - o(1))\log(1/\eps)$.

\item \label{item:state-space-lower-bounds-2}
If $\mathrm{Bernoulli}(1/(2^m+1)) \in \DAlphabet$ for infinitely many $m \in \N$,
then $\abs{\AState} \geq 1/\eps$.
\qedhere
\end{enumerate}
\end{theorem}
\end{nicebox}

\Cref{theorem:space-lb} is stronger than the original
conjecture: it shows that a space lower bound holds even if the only output
distribution handled by the sampler is $\mathrm{Bernoulli}(1/3)$.
By slightly enlarging the set of distributions to include
Bernoulli distributions with rational probability $1/(2^m+1)$, we can
obtain a tighter space lower bound of $\ceil{\log(1/\eps)}$ bits.

We also establish an analogous space lower bound for random sampling with
i.i.d.~outputs, which is the setting studied by \citet{kozen2022}.
\begin{nicebox}
\begin{theorem}[store=theorem:space-lb-iid,restate-keys={note=Restated}]
\label{theorem:space-lb-iid}
Let $k \geq 2$ be an integer.
For almost all discrete distributions $p$ over $k$ outcomes,
any sampler generating i.i.d.~samples from $p$ with an entropy loss bound of $\eps$ bits per sample
requires at least $(1/(k+1)-o(1))\log(1/\eps)$ bits of persistent space.
\end{theorem}
\end{nicebox}
In the theorem statement, ``almost all'' refers to Lebesgue-almost every
point $p \in (0,1)^{k-1}$ with $\sum_{i=1}^{k-1}p_i < 1$, under the standard
parameterization $p \mapsto (p_1, \dots, p_{k-1}, 1-\sum_{i=1}^{k-1}p_i)$ of
the $(k-1)$-dimensional probability simplex.
Specializing \cref{theorem:space-lb-iid} to $k=2$, we obtain a lower bound of
$(1/3-o(1))\log(1/\eps)$ bits for almost every Bernoulli
distribution, improving the constant $1/5.116201$ in
\zcref[noname]{item:state-space-lower-bounds-1} from
\cref{theorem:space-lb}.
Moreover, \cref{theorem:space-lb-iid} applies to a less stringent setting
than \zcref[noname]{item:state-space-lower-bounds-2} from
\cref{theorem:space-lb}: the former need only
generate i.i.d.~samples from a fixed distribution $p$ with entropy loss at most $\eps$,
while the latter must handle infinitely many Bernoulli distributions.

\paragraph{Organization}
The rest of this paper is organized as follows.
\Cref{sec:background} reviews existing algorithms for random sampling,
illustrates the tradeoff between entropy efficiency and space complexity,
and provides intuition for our information-theoretic lower bounds.
\Cref{sec:preliminaries} formally defines online random samplers and entropy loss.
We then introduce two new concepts used to prove
the $\Omega(\log(1/\eps))$ lower bound on the size of the state space.
In \cref{sec:information}, we define an abstract
``state information content'' function for an arbitrary sampler
and show how to construct this function based only on the sampler's
transition function.
In \cref{sec:cycles}, we prove a very general bound showing the state-space
size needed to ensure an entropy loss of at most $\eps$ bits per sample
when generating i.i.d.~samples from a fixed distribution.
\Cref{sec:bernoulli} applies this bound to Bernoulli distributions
to prove \cref{theorem:space-lb}.
\Cref{sec:iid} applies the same bound, together with Diophantine approximation
on manifolds, to prove \cref{theorem:space-lb-iid}.
\Cref{sec:related-work} discusses related work in random sampling and coding theory.
\Cref{sec:open-questions} concludes with open questions.

\section{Background on Random Sampling}
\label{sec:background}

To motivate our lower bounds, we first present a simple example of a
random sampler that maintains persistent space.
This example demonstrates both the entropy-space tradeoff and the concept
of state information content, which is a main tool in our proof of
\cref{theorem:space-lb}.

\subsection{A Stateless Sampling Algorithm}
\label{sec:background-stateless}

Consider the problem of converting
an i.i.d.~sequence $(\Flip_1, \Flip_2, \Flip_3, \dots)$ of uniform bits
into an i.i.d.~sequence $(X_1, X_2, X_3, \ldots)$ of uniform trits;
i.e.,
$\Flip_i \sim \mathrm{Uniform}\set{0,1}$ and $X_i \sim \mathrm{Uniform}\set{0,1,2}$.
The standard inversion method~\citep[\S2]{devroye1986} generates $X_1$ by
first partitioning the unit interval
$J_1 \defeq [0,1]$ into three equally sized subintervals:
\begin{align}
[0,1] = \underbrace{\left[0,\textstyle\frac13\right]}_{J_{10}}
   \; \bigcup \; \underbrace{\textstyle\left[\frac13, \frac23\right]}_{J_{11}}
   \; \bigcup \; \underbrace{\textstyle\left[\frac23,1\right]}_{J_{12}}.
\label{eq:interval-partition-J1}
\end{align}
It then draws a uniform number $U \sim \mathrm{Uniform}[0,1]$,
and returns $X_1 \in \set{0,1,2}$ such that $U \in J_{1X_1}$.
To implement this algorithm using an entropy source that provides random bits,
we recognize that the expansion $U = (0.\Flip_1 \Flip_2 \Flip_3 \dots )_2$
satisfies $\Flip_i \sim \mathrm{Uniform}\set{0,1}$.
We keep drawing input bits until we can
determine whether $U \in J_{10}$, $U \in J_{11}$, or $U \in J_{12}$,
and set $X_1$ accordingly.

The top row of \cref{fig:interval} shows an example of this algorithm in
action.
The random bits $\Flip_1\Flip_2\Flip_3 = \texttt{100}$ let us determine that
$U \in J_{11}$, so $X_1 = 1$.
Generating an output in this manner costs 3 input bits on average,
whereas the output trit has $\log(3) \approx 1.585$ bits of entropy
and an optimal algorithm~\citep{knuth1976} uses ${\approx}\,2.667$ bits.
To generate $n$ i.i.d.~outputs $X_1, \dots, X_n$,
a stateless sampler repeats this process $n$ times from scratch,
consuming a total of $3n$ random bits on average.

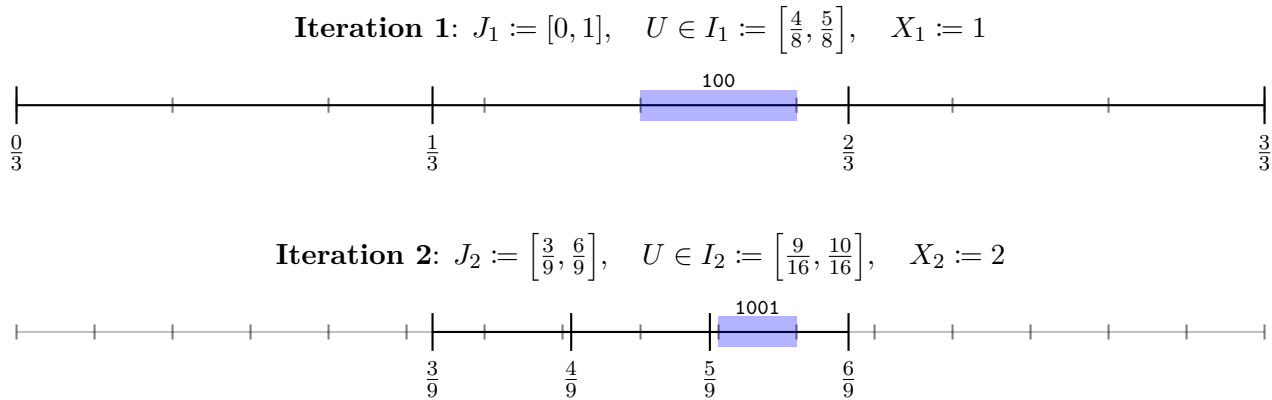
\begin{figure}[t]
\centering
\def\scale{\linewidth}
\begin{tikzpicture}[thick,scale=1]
\draw[-] (0,0) -- (\scale,0)
   node[pos=0.5,yshift=10mm,anchor=center,draw=none]{
      \textbf{Iteration 1}: $\begin{aligned}
      J_1 \defeq \left[0,1\right]\!,\quad
      U \in I_1 \defeq \textstyle\left[\frac48, \frac58\right]\!, \quad
      X_1 \defeq 1
      \end{aligned}$}
      ;
\foreach \x in {0, 1, 2, 3} {
   \draw[-] (\scale*\x/3, -0.25) -- (\scale*\x/3,0.25) node[pos=0,below,font=\normalsize]{$\frac{\x}{3}$};
}
\foreach \x in {0, ..., 8} {
   \draw[-,opacity=0.5] (\scale*\x/8, -0.1) -- (\scale*\x/8,0.1) node[pos=0,below,font=\normalsize]{};
}
\draw[thick,color=blue,fill=blue,opacity=0.3,draw=none]
   (\scale*4/8, 0.2)
   rectangle (\scale*5/8,-0.2)
   node[pos=0.5,anchor=south,color=black,opacity=1,yshift=1mm,font=\ttfamily\scriptsize]{100};

\def\ypos{-3}
\draw[draw=none] (0,\ypos) -- (\scale,\ypos)
   node[pos=0.5,yshift=10mm,anchor=center]{
      \textbf{Iteration 2}: $\begin{aligned}
      J_2 \defeq \textstyle\left[\frac39, \frac69\right]\!,\quad
      U \in I_2 \defeq \textstyle\left[\frac{9}{16}, \frac{10}{16}\right]\!,\quad
      X_2 \defeq 2
      \end{aligned}$}
      ;

\draw[-,opacity=0.25] (0,\ypos) -- (1/3*\scale,\ypos);
\draw[-,opacity=1] (1/3*\scale,\ypos) -- (2/3*\scale,\ypos);
\draw[-,opacity=0.25] (2/3*\scale,\ypos) -- (3/3*\scale,\ypos);

\foreach \x in {0, ..., 3} {
   \pgfmathsetmacro{\opacity}{\x == 0 || \x == 3 ? 0 : 1}
   \pgfmathtruncatemacro{\xx}{3*\x}
   \draw[-,opacity=\opacity] ([yshift=-0.25cm]\scale*\x/3, \ypos) -- ([yshift=0.25 cm]\scale*\x/3,\ypos)
      node[pos=0,below,font=\normalsize]{$\frac{\xx}{9}$};
}

\coordinate (P) at ($(\scale * 1/3,\ypos) ! {1/3} ! (\scale * 2/3,\ypos)$);
\coordinate (Q) at ($(\scale * 1/3,\ypos) ! {2/3} ! (\scale * 2/3,\ypos)$);
\draw[-,draw] ([yshift=-0.25 cm]P) -- ([yshift=0.25 cm]P) node[pos=0,below,font=\normalsize]{$\frac{4}{9}$};
\draw[-,draw] ([yshift=-0.25 cm]Q) -- ([yshift=0.25 cm]Q) node[pos=0,below,font=\normalsize]{$\frac{5}{9}$};

\draw[thick,color=blue,fill=blue,opacity=0.3,draw=none]
   ([yshift=0.2cm]\scale*9/16, \ypos)
   rectangle
   ([yshift=-0.2cm]\scale*10/16, \ypos)
   node[pos=0.5,anchor=south,color=black,opacity=1,yshift=1mm,font=\ttfamily\scriptsize]{1001};

\foreach \x in {0,...,16} {
   \draw[-,opacity=.5]
      ([yshift=-0.1cm]\scale*\x/16, \ypos) -- ([yshift=0.1cm]\scale*\x/16, \ypos)
      node[pos=0,below,font=\normalsize]{};
}
\end{tikzpicture}
\caption{Two iterations of the interval algorithm~\citep{han1997} for
converting i.i.d.~flips $U$ of a fair two-sided coin into i.i.d.~rolls $X_1,
X_2, \dots$ of a fair three-sided ``dice''.
The state of the sampler at the end of iteration $n$
is the dyadic interval $I_n = [a_n/2^k, (a_n+1)/2^k]$
determined by the $k$ coin
tosses $U_1,\dots,U_k$ used so far, and the triadic interval
$J_n = [b_n/3^n, (b_n+1)/3^n]$ that was just split.
}
\label{fig:interval}
\end{figure}

\subsection{Reducing Entropy Cost by Maintaining State}
\label{sec:background-stateful}

\Citet{han1997}
introduce a random sampling method called the interval algorithm that
reduces the entropy loss of na\"ive inversion sampling by maintaining state
across iterations.
The key idea is shown in the second row of \cref{fig:interval}.
After generating the first output $X_1=1$, we have determined that
$U \in I_1 \defeq [4/8, 5/8] \subset [1/3,2/3] = J_{11}$.
Rather than resample from scratch,
which would discard the three bits used so far,
we split $J_{11} \eqdef J_2$ into three new subintervals
\begin{align}
\textstyle\left[\frac13,\frac23\right] = \underbrace{\textstyle\left[\frac39,\frac49\right]}_{J_{2,0}}
   \; \bigcup \; \underbrace{\textstyle\left[\frac49, \frac59\right]}_{J_{2,1}}
   \; \bigcup \; \underbrace{\textstyle\left[\frac59,\frac69\right]}_{J_{2,2}}.
\label{eq:interval-partition-J2}
\end{align}
We then draw fresh bits until we can determine that $U \in J_{2,X_2}$.
This approach effectively reuses the first three bits $\texttt{100}$ drawn
in the first iteration, by using the fact that $U \mid U \in [4/8,5/8]$
follows a uniform distribution over $[4/8, 5/8]$.
In the example from \cref{fig:interval}, when $\Flip_4 = \texttt{1}$
we can determine that $X_2 = 2$ using just \textit{one} extra bit, which
is impossible using the stateless sampler.
More generally, the interval algorithm requires 5 bits on average to produce two
outputs $X_1, X_2$, instead of the 6 bits used by the stateless sampler.

The entropy loss decreases rapidly with $n$; for $n=10$ the stateful algorithm
uses $17.8$ bits on average, which is close to the information-theoretically
optimum of roughly $16.4$ bits using the method of \citet{knuth1976}
and much smaller than the $30$ bits incurred by the stateless algorithm.
However, this improvement in entropy incurs a cost in terms of space
complexity.
The expected number of bits needed to represent the state $(I_n, J_n)$
of the interval algorithm grows as $\Theta(n)$, which is unbounded as $n \to \infty$
(a detailed analysis of the space and time complexity of the \citep{han1997}
interval method is given in \citet[Appendix A]{draper2026soda}).

\subsection{State Spaces of Existing Random Sampling Algorithms}
\label{sec:background-information-content}

Many entropy-efficient algorithms in the literature exhibit a similar
entropy-space tradeoff to the one described in \cref{sec:background-stateful}.
A key idea for our space lower bounds is to analyze the information
contained in the state of an online random sampler.
Recall that the information content of a random variable $X$ drawn from a
distribution $p$ is defined as $\log(1/p(X))$ \citep{shannon1948}.
For example, each outcome of a fair coin flip has
$\log(1/0.5) = 1$ bit of information, and an outcome with $p(X) = 1/3$ has
$\log(3) \approx 1.585$ bits of information.
We give a brief survey of the information stored in the state of various
sampling algorithms, before giving a formal definition
in \cref{sec:information}.

\paragraph{Interval Algorithm~\citep{han1997}}
For the method described in \cref{sec:background-stateful},
given a state $I = [I_l, I_r] \subseteq [J_l, J_r] = J$,
the total number of bits consumed so far is given by $\log(1/(I_r - I_l))$,
and the information content of the outputs so far is $\log(1/(J_r - J_l))$ bits
(for the case of all output trits in \cref{fig:interval},
$\log(1/(J_r - J_l))$ is the same as the number of outputs multiplied by $\log(3)$).
Because the output information cannot exceed the input information at any time,
$\log(1/(I_r - I_l)) - \log(1/(J_r - J_l)) \geq 0$ provides an upper bound on the
information content which the algorithm can extract from the state for future samples.
Hence, we call $\log({J_r-J_l}) - \log({I_r-I_l})$ the \textit{state information content}.

\paragraph{Randomness Recycling~\citep{draper2026soda}}
The online random sampling algorithms of \citet{draper2026soda}
maintain a state consisting of
a pair of integers $(Z,M)$, where $M \geq 1$ and
$Z$ is uniformly distributed over $\set{0,1,\ldots,M-1}$.
Therefore, the information content of the state $(Z,M)$ is
given by the information content of $Z$, namely $\log M$,
because each outcome has probability $1/M$.
These uniform states were used by \citet{jacques2004} for the special case
of generating a sequence of discrete uniform random variables using
i.i.d.~random bits.
\Cref{fig:states-rr} shows an example.
A uniform integer $U$ is drawn between 0 and 15.
If the pink cell with label 2 is selected, then $\mathrm{A}$ is returned.
The conditional distribution of $U$ is uniform over $\set{0,\ldots,4}$;
so the uniform state is $(Z\defeq 2, M\defeq 5)$
with $Z \mid M \sim \mathrm{Uniform}\set{0,\ldots,M-1}$.


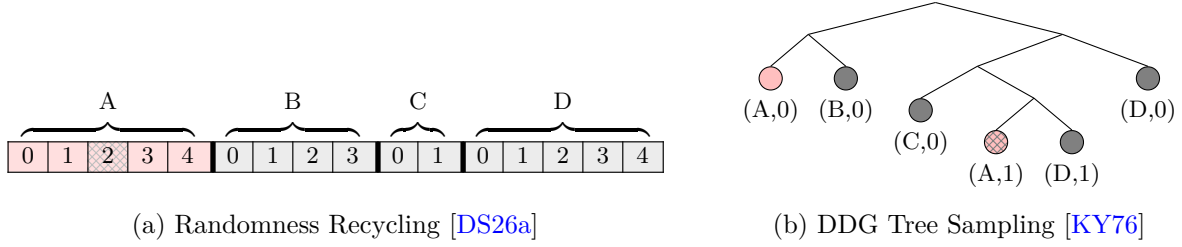
\begin{figure}[!t]
\begin{subfigure}[b]{.6\linewidth}
\centering
\footnotesize
\newcommand{\clrA}{pink!50!white}
\newcommand{\clrB}{gray!15!white}
\newcommand{\clrC}{gray!15!white}
\newcommand{\clrD}{gray!15!white}
\newcommand{\cA}[1][\mathrm{A}]{\cellcolor{\clrA}#1}
\newcommand{\cB}[1][\mathrm{B}]{\cellcolor{\clrB}#1}
\newcommand{\cC}[1][\mathrm{C}]{\cellcolor{\clrC}#1}
\newcommand{\cD}[1][\mathrm{D}]{\cellcolor{\clrD}#1}
\begin{align*}
\NiceMatrixOptions{custom-line={letter = I, tikz = {line width=2pt, black }}}
\begin{NiceArray}[hvlines]{cccccIccccIccIccccc}
\cA[0] & \cA[1] & \Block[tikz={pattern=crosshatch, pattern color=gray!50!white, preaction={fill=pink!50!white}}]{}{2} & \cA[3] & \cA[4] &
\cB[0] & \cB[1] & \cB[2] & \cB[3] &
\cC[0] & \cC[1] &
\cD[0] & \cD[1] & \cD[2] & \cD[3] & \cD[4]
\CodeAfter
  \OverBrace[shorten,yshift=4pt]{1-1}{1-5}{\mathrm{A}}
  \OverBrace[shorten,yshift=4pt]{1-6}{1-9}{\mathrm{B}}
  \OverBrace[shorten,yshift=4pt]{1-10}{1-11}{\mathrm{C}}
  \OverBrace[shorten,yshift=4pt]{1-12}{1-16}{\mathrm{D}}
\end{NiceArray}
\end{align*}
\caption{Randomness Recycling~\citep{draper2026soda}}
\label{fig:states-rr}
\end{subfigure}\hfill
\begin{subfigure}[b]{.4\linewidth}
\centering
\footnotesize
\tikzset{every tree node/.style={anchor=north}}
\tikzset{level distance=12pt}
\tikzset{sibling distance=1pt}
\begin{tikzpicture}
\Tree[
  [
    \node[label={below:(A,0)},draw=black,circle,fill=pink]{};
    \node[label={below:(B,0)},draw=black,circle,fill=gray]{};
  ]
  [
    [
      \node[label={below:(C,0)},draw=black,circle,fill=gray]{};
      [
        \node[label={below:(A,1)},draw=black,circle,preaction={fill=pink},pattern=crosshatch,pattern color=gray]{};
        \node[label={below:(D,1)},draw=black,circle,fill=gray]{};
      ]
    ]
    \node[label={below:(D,0)},draw=black,circle,fill=gray]{};
  ]
]
\end{tikzpicture}
\caption{DDG Tree Sampling~\citep{knuth1976}}
\label{fig:states-ddg}
\end{subfigure}

\caption{Examples of random state extraction when using randomness recycling and DDG tree sampling,
for a discrete distribution
$\set{\mathrm{A} \mapsto 5/16, \mathrm{B}\mapsto 4/16, \mathrm{C} \mapsto 2/16, \mathrm{D} \mapsto 5/16}$.
}
\label{fig:states}
\end{figure}

\paragraph{DDG Tree Samplers~\citep{knuth1976}}
Many sampling algorithms explicitly traverse a data structure called the
discrete distribution generating (DDG)
tree~\citep{knuth1976,saad2025,draper2026ieee,saad2020fldr,roy2013}.
The online sequential version of the algorithm maintains a state consisting
of a binary tree, together with the current traversal position (which is a
leaf node, after producing a sample).
As noted by \citet[\S2.3.2]{devroye2020}, the usable information in such a DDG state
is given by the distribution of leaves in the tree that share the same label
as the leaf at the current position
(and this information can be reduced to an interval state by
fixing an ordering among the leaves sharing the same label).
Therefore, the state information content can be defined as the negative log-probability
of the leaf, under the distribution of leaves that share the same output label.
\Cref{fig:states-ddg} shows an example.
A random leaf is selected with log-probability proportional to its depth.
If leaf $(\mathrm{A},1)$ is selected, then $\mathrm{A}$ is returned.
The conditional distribution over the $\mathrm{A}$ leaves is
$M \defeq \set{0 \mapsto 4/5, 1 \mapsto 1/5}$;
so the nonuniform state is $(Z \defeq 1, M)$
with $Z \mid M \sim M$.

\paragraph{Offline Batched Sampling}
Lastly, consider the problem where all target distributions
$p_1, p_2, \ldots$ are known in advance,
and the samples are generated in batches to improve entropy efficiency
(e.g., by applying the interval algorithm to the joint distribution
$p_{1:n}(x_1,\dots,x_n) \defeq p_1(x_1)\dots p_n(x_n)$
of a batch of $n$ samples),
while no state is maintained between batches.
In the model where outputs can only be emitted one at a time, the sampler generates
$X_1,\ldots,X_n$ but only emits $X_1$ in the first round;
the state stored for the second round is simply the remaining samples $X_2, \ldots, X_n$,
which have information content $\log(1/p(X_2, \ldots, X_n))$ bits.
In general, the state information content of this batched sampler is
the sum of the information content of the samples which have been generated but not yet emitted.
Although batched samplers with knowledge of future distributions do not satisfy our
definition of online random samplers in \cref{def:online-random-sampling},
for the special case of generating i.i.d.~samples
(i.e., there is only one possible target distribution),
our lower bounds apply.

\subsection{Proof Idea}

By maintaining state, an online random sampler can store information that
has been consumed from the input source but not yet emitted in an output sample.
Suppose a sampler is at state $s$, reads $r$ input bits, produces output
$x \sim p$, and moves to a new state $s'$.
The information loss of this transition is the
difference between the consumed and produced information:
\begin{equation}
0 \le \left(r + h(s)\right) - \left(I_p(x) + h(s')\right),
\end{equation}
where $h(s)$ and $h(s')$ denote the information content of states $s$ and $s'$,
and $I_p(x) \defeq \log(1/p(x))$ is the information content of the output sample.
For the loss to be small, say at most $\eta$, we need
\begin{equation}
h(s') \in [h(s)+r-I_p(x)-\eta, h(s)+r-I_p(x)].
\end{equation}
As the number $r$ of input bits is an integer, the \textit{fractional}
state information content must satisfy
\begin{equation}
\modZ{h(s')} \in \set{\modZ{v} \mid v \in [h(s)-I_p(x)-\eta, h(s)-I_p(x)]},
\end{equation}
where $\modZ{x}$ is the fractional part of a real number $x$.
A transition that achieves a small loss must therefore choose a next state $s'$
whose fractional information content lies just before the ideal position
$\modZ{h(s)-I_p(x)}$ that would achieve zero loss; the gap between the two
becomes permanently lost information.
Therefore, to guarantee a small loss, the possible states of the algorithm
should have information content whose fractional parts are densely spaced in $[0,1)$.

One strategy is to use states with information content uniformly spaced
over the unit interval.
\Cref{fig:demo} illustrates this idea
for $S=5$ states $s_0,\ldots,s_4$ with fractional information content
$\modZ{h(s_i)} = i/5$ for $i=0,\dots,4$.
By choosing an appropriate new state after each sample, such a sampler
could theoretically ensure that fewer than $\eta \defeq 1/S=0.2$ bits are lost
after each sample.


\begin{figure}[H]
\centering
\definecolor{cbBlue}{HTML}{0072B2}
\definecolor{cbRed}{HTML}{D55E00}
\tikzset{ckpt/.style={anchor=center,circle,minimum size=6pt,fill=black,inner sep=0pt,}}
\begin{tikzpicture}[thick]
\draw[-, line width=1pt] (0,0) -- (14,0)
  node[name=ha,pos=0.0,ckpt,label={below:$s_0$}]{}
  node[name=ia,pos=0.17,ckpt,draw=none,fill=none]{}
  node[name=hb,pos=0.2,ckpt,label={below:$s_1$}]{}
  node[name=ib,pos=0.37,ckpt,draw=none,fill=none]{}
  node[name=hc,pos=0.4,ckpt,label={below:$s_2$}]{}
  node[name=ic,pos=0.57,ckpt,draw=none,fill=none]{}
  node[name=hd,pos=0.6,ckpt,label={below:$s_3$}]{}
  node[name=id,pos=0.77,ckpt,draw=none,fill=none]{}
  node[name=he,pos=0.8,ckpt,label={below:$s_4$}]{}
  node[name=ie,pos=0.97,ckpt,draw=none,fill=none]{}
  node[name=hf,pos=1.0,ckpt,label={below:$s_0$}]{}
  node[name=zero,pos=0.0,ckpt,draw=none,fill=none,label={left:$0$}]{}
  node[name=one,pos=1.0,ckpt,draw=none,fill=none,label={right:$1$}]{}
  ;

\node[draw=none, fill=cbBlue, opacity=0.25, thick, inner sep=0pt, fit=(ia.south) (hb.north)] {};
\node[draw=none, fill=cbRed,   opacity=0.25, thick, inner sep=0pt, fit=(ha.south) (ia.north)] {};
\node[draw=none, fill=cbBlue, opacity=0.25, thick, inner sep=0pt, fit=(ib.south) (hc.north)] {};
\node[draw=none, fill=cbRed,   opacity=0.25, thick, inner sep=0pt, fit=(hb.south) (ib.north)] {};
\node[draw=none, fill=cbBlue, opacity=0.25, thick, inner sep=0pt, fit=(ic.south) (hd.north), label={[font=\scriptsize]above:$\alpha$}] {};
\node[draw=none, fill=cbRed,   opacity=0.25, thick, inner sep=0pt, fit=(hc.south) (ic.north)] {};
\node[draw=none, fill=cbBlue, opacity=0.25, thick, inner sep=0pt, fit=(id.south) (he.north)] {};
\node[draw=none, fill=cbRed,   opacity=0.25, thick, inner sep=0pt, fit=(hd.south) (id.north)] {};
\node[draw=none, fill=cbBlue, opacity=0.25, thick, inner sep=0pt, fit=(ie.south) (hf.north)] {};
\node[draw=none, fill=cbRed,   opacity=0.25, thick, inner sep=0pt, fit=(he.south) (ie.north)] {};

\draw[-latex] (hb.south) to[out=210,in=330] (ha.south);
\draw[-latex] (hc.south) to[out=210,in=330] (hb.south);
\draw[-latex,dashed] (hd.south) to[out=210,in=330] node[pos=0.5,yshift=-.25cm,font=\footnotesize]{$1/S$} (hc.south);
\draw[-latex] (he.south) to[out=210,in=330] (hd.south);
\draw[-latex] (hf.south) to[out=210,in=330] (he.south);

\draw[decorate,black,decoration={brace,raise=.275cm}] (hc.center) -- (ic.center) node[pos=0.5,yshift=0.6cm,font=\footnotesize]{information loss};

\end{tikzpicture}
\caption{Visualization of the information lost during a state transition of a sampler.}
\label{fig:demo}
\end{figure}
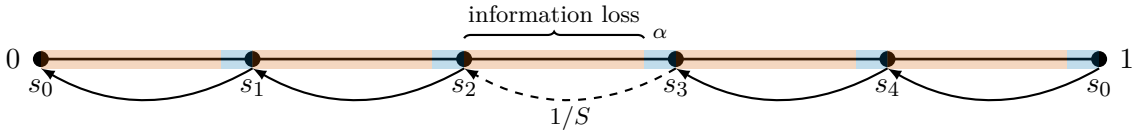

We construct distributions that \textit{force} the sampler to lose nearly
$1/S$ bits of information on average.
Consider a low-entropy $\mathrm{Bernoulli}(\delta)$
distribution, for some small $\delta > 0$.
When the sampler generates the likely outcome $X=0$,
the output sample has information content
$\alpha \defeq \log(1/(1-\delta)) = \Theta(\delta)$,
which is small but positive.
So if the sampler retained the same state $s' = s$, it would need to consume
at least one fresh bit to produce the output, losing $1-\alpha$ bits.
Otherwise, the sampler must transition to a new state $s' \ne s$,
and the gap of $1/S$ between states means that at least $1/S-\alpha$ bits will be lost.
\Cref{fig:demo} illustrates the $\alpha$ bits of output information using blue shading,
and the $1/S-\alpha$ bits of information lost over a transition is shown using orange shading,
with a label for the $s_3 \to s_2$ information loss.
Thus, sampling the likely outcome $X=0$ contributes
$(1-\delta)(1/S-\alpha)$ to the average information loss,
which approaches $1/S$ as $\delta \to 0$.

In the more general case where the state information content is not uniformly spaced,
we see that the blue intervals remain the same,
while the orange intervals may shrink or expand,
so certain transitions may become more efficient and others less efficient.
However, the overall loss after $S$ transitions remains $1-S\alpha$
because the total length of the orange intervals is invariant.
Thus, the loss per sample remains $1/S$, so we need $S \geq 1/\eps$
to ensure that the loss per sample is at most $\eps$.

In the remainder of this paper, we formalize this intuition to establish
\cref{theorem:space-lb}, and then generalize it to prove lower bounds for
even weaker samplers with i.i.d.~outputs to prove \cref{theorem:space-lb-iid}.

\section{Preliminaries}
\label{sec:preliminaries}

We now formally define online random samplers and information loss.

\begin{notation}
The base-two logarithm is denoted by $\log$.
The fractional part of a real number $x$ is denoted by $\modZ{x} \defeq x - \floor{x}$.
The set of finite binary strings is denoted by $\set{0,1}^*$,
and the length of a binary string $\mathbf{c} \in \set{0,1}^*$ is $\abs{\mathbf{c}}$.
The set of natural numbers is $\N \defeq \set{0,1,\dots}$.
The support of a distribution $p$ is $\supp(p) \defeq \set{x \in \N \mid p(x) > 0}$.
The greatest lower bound of a set $S \subset \R$ is $\inf S$,
which is defined to be $+\infty$ if $S$ is empty.
\end{notation}

\begin{definition}
\label{def:online-random-sampling}
Let $\DAlphabet$ be a set of discrete probability distributions over
$\N \defeq \set{0,1,\dots}$.
Let $\AState$ be a set of states and $s_0 \in \AState$ a designated initial state.
An \textit{online random sampler} is a partial function
\begin{equation}
f : \DAlphabet \times \AState \times \set{0,1}^* \rightharpoonup \N \times \AState
\end{equation}
such that the following conditions hold.
\begin{itemize}
\item \textit{Prefix Free}:
For every $p \in \DAlphabet$ and $s \in \AState$,
if
$(p, s, \mathbf{c}) \in \dom(f)$
then
$(p, s, \mathbf{c}') \notin \dom(f)$
for every $\mathbf{c}, \mathbf{c}' \in \set{0,1}^*$
such that $\mathbf{c}$ is a proper
prefix of $\mathbf{c}'$.

\item \textit{Exhaustive}:
For every $p \in \DAlphabet$ and $s \in \AState$,
we have
$\sum_{\substack{\mathbf{c} \mid (p,s,\mathbf{c}) \in \dom(f)}}
   2^{-\abs{\mathbf{c}}} = 1$.

\item \textit{Exact}: For all distribution sequences
$p_1, p_2, \ldots \in \DAlphabet$
and output sequences $x_1, x_2, \ldots \in \N$
and for each finite prefix length $n \in \N$, we have
\begin{equation}
\sum_{\substack{(s_1, \mathbf{c}_1),\ \dots,\ (s_n, \mathbf{c}_n) \\
                \mbox{s.t. }\forall i \in \set{1, \ldots, n}, \;
                f(p_i, s_{i-1}, \mathbf{c}_i) = (x_i, s_i)}}
\left( 2^{-\sum_{i=1}^n \abs{\mathbf{c}_i}} \right)
= \prod_{i=1}^n p_i(x_i).
\qedhere
\label{eq:online-random-sampling-exact}
\end{equation}
\end{itemize}
\end{definition}

\begin{remark}
To avoid states with undefined behavior,
we assume without loss of generality that each state $s$ is
reachable from the initial state $s_0$, i.e.,
there exist distributions $p_1, \ldots, p_k \in \DAlphabet$,
outputs $x_1, \ldots, x_k \in \N$,
input bit sequences $\mathbf{c}_1, \ldots, \mathbf{c}_k \in \set{0,1}^*$,
and states $s_1, \ldots, s_k = s$ such that
$f(p_i, s_{i-1}, \mathbf{c}_i) = (x_i, s_i)$ for each $i \in \set{1, \ldots, k}$.
This condition is equivalent to restricting $\AState$
to be the subset of states reachable from $s_0$.
\end{remark}

\Cref{def:online-random-sampling} generalizes the
formulation in \citet[Listing 1.1]{draper2026soda}.
As we are not concerned with the computational complexity at sampling time,
we allow noncomputable dependence of the sampler on the distributions.
With this definition, our lower bounds apply in full generality, with no
assumptions on how irrational probabilities are represented.

\begin{definition}
\label{def:information-loss}
The \textit{information loss} of an online random sampler $f$
when sampling $x_1, \ldots, x_n$ from target distributions $p_1, \ldots, p_n$
using input bit sequences $\mathbf{c}_1, \ldots, \mathbf{c}_n$ is
\begin{equation}
\sum_{i=1}^n \left( \abs{\mathbf{c}_i} - \log\frac{1}{p_i(x_i)} \right).
\end{equation}
The \textit{entropy loss} of $f$ when sampling from $p_1, \ldots, p_n$ is
\begin{equation}
\sum_{\substack{(s_1,x_1,\mathbf{c}_1),\ \ldots,\ (s_n,x_n,\mathbf{c}_n) \\
                \mbox{s.t. }\forall i \in \set{1, \ldots, n}, \;
                f(p_i, s_{i-1}, \mathbf{c}_i) = (x_i, s_i)}}
   \left( \sum_{i=1}^n \abs{\mathbf{c}_i} - \log{\frac{1}{p_i(x_i)}} \right)
   2^{-\sum_{i=1}^{n}\abs{\mathbf{c}_i}},
\end{equation}
which can be understood as the expected information loss given uniformly random input bits.
The sampler $f$ is said to achieve an
\textit{entropy loss bound} of $\eps$ bits per sample if
for all $p_1,p_2,\ldots \in \DAlphabet$,
the entropy loss of $f$ when sampling from $p_1, \ldots, p_n$
is at most $\eps n + o(n)$ as $n \to \infty$.
\end{definition}

\section{State Information Content}
\label{sec:information}

An online random sampler maintains some state between each iteration.
We can analyze the information content of the state given a description
of how the sampler operates, as in the interval method from \cref{sec:background}.
However, to prove general lower bounds on the state space, we need a way
to analyze the information content of a state without
knowing the inner workings of the sampler.
The key fact we need about the state information content is that it gives
an upper bound on how much information the sampler can extract from the
stored state for future samples, in place of drawing fresh random bits from
the entropy source.
Therefore, we can define the information content of a state
for an arbitrary online random sampler as follows.

\begin{definition}
\label{def:state-information-content}
Let $f$ be an online random sampler with state space $\AState$
and target distribution set $\DAlphabet$.
The \textit{state information content} is a function
$h : \AState \to \R_{\geq 0}$ such that $h(s_0) = 0$, and
for each pair of states $s, s' \in \AState$,
\begin{equation}
h(s') + \log(1/p(x)) \leq h(s) + \abs{\mathbf{c}},
\end{equation}
whenever $f(p, s, \mathbf{c}) = (x, s')$ for some target distribution $p \in \DAlphabet$,
some $x$ in the support of $p$,
and some finite sequence of input bits $\mathbf{c}$.
\end{definition}

This definition is similar to the data processing inequality,
because the input state information is independent of the input bits
and the output state information is independent of the output sample.
The definition implies that any state information content function
must satisfy a system of difference constraints,
which we can compactly express in an $\abs{\AState} \times \abs{\AState}$ matrix:
\begin{equation}
\label{eq:difference-constraints}
h(s_j) - h(s_i) \leq A_{i, j},
\end{equation}
where
\begin{equation}
\label{eq:difference-constraints-matrix}
A_{i, j} \defeq
\inf \set*{ \abs{\mathbf{c}} - \log\frac{1}{p(x)}
   \; \Bigg\vert\;
   p \in \DAlphabet,
   x \in \N,
   \mathbf{c} \in \set{0,1}^*,
   f(p, s_i, \mathbf{c}) = (x, s_j)
   }.
\end{equation}

Informally, $A_{i,j}$ is the minimal possible marginal information loss incurred
when transitioning from state $s_i$ to $s_j$,
including any extra information put into the state.
When $s_i$ contains more information than $s_j$, we may have $A_{i,j} < 0$,
but the quantity $A_{i,j} - (h(s_j) - h(s_i)) \geq 0$
(i.e., slack in \cref{eq:difference-constraints})
will represent permanently lost information
that is contained in neither the output $x \sim p$ nor the next state $s_j$.
To find a solution to this system of difference constraints
\cref{eq:difference-constraints},
we can apply classical techniques used in graph algorithms.

\begin{theorem}[Feasibility of difference constraints \citep[Theorem 22.9]{clrs2022}]
\label{theorem:difference-constraints-feasibility}
Let $G$ be the directed graph with vertex set $\AState$
and edge weights $A_{i, j}$ between each pair of vertices $s_i, s_j$.
If $G$ contains no negative-weight cycles, then
a feasible solution to the system \cref{eq:difference-constraints}
is given by
\begin{equation}
h(s_i) = \delta(s_0, s_i),
\end{equation}
where $s_0$ is the initial state and $\delta(s_0, s_i)$ is
the length of the shortest path from $s_0$ to $s_i$ in $G$.
If $G$ contains a negative-weight cycle, then
there is no feasible solution for the system.
\end{theorem}

A shortest path from $s_0$ to each $s_i$ in $G$ can be computed using the
Bellman-Ford algorithm.
We rule out the existence of negative-weight cycles in $G$ by analyzing
the probability of input and output sequences corresponding to such a cycle.
A negative-weight cycle would imply that the sampler generates a certain
output sequence with greater probability than it correctly should.

\begin{theorem}[Nonexistence of negative cycles]
\label{theorem:no-negative-cycles}
Let $G$ be the directed graph with vertex set $\AState$
and edge weights $A_{i, j}$ between each pair of vertices $s_i, s_j$.
There are no negative-weight cycles in $G$.
\end{theorem}

\begin{proof}
Assume, for the sake of contradiction, that there is a negative-weight cycle
$(i_0, i_1, \dots, i_k = i_0)$
in $G$, i.e.,
$A_{i_0, i_1} + A_{i_1, i_2} + \cdots + A_{i_{k-1}, i_k} < 0$.
Define
\begin{equation}
\epsilon \defeq - (A_{i_0, i_1} + A_{i_1, i_2} + \cdots + A_{i_{k-1}, i_k}) / k.
\label{eq:epsilon}
\end{equation}
For each $j \in \set{1, 2, \ldots, k}$, let $p_j \in \DAlphabet$ and $x_j \in \N$
and $\mathbf{c}_j \in \set{0,1}^*$ satisfy
\begin{equation}
f(p_j, s_{i_{j-1}}, \mathbf{c}_j) = (x_j, s_{i_j}),
\qquad
\abs{\mathbf{c}_j} - \log(1/p_j(x_j)) < A_{i_{j-1}, i_j} + \epsilon.
\label{eq:bad-cycle}
\end{equation}
Such values exist by the definition of $A_{i, j}$ as an infimum
in \cref{eq:difference-constraints-matrix}.
Because state $s_{i_0}$ is reachable from $s_0$ by \cref{def:online-random-sampling}, we may
let $(q_1,\ldots,q_l)$, $(y_1,\ldots,y_l)$, $(\mathbf{b}_1,\ldots,\mathbf{b}_l)$,
and $(r_0=s_0, r_1, \ldots, r_{l-1}, r_l=s_{i_0})$ be sequences of
distributions, samples, input bit sequences, and states, respectively, such that
$f(q_j, r_{j-1}, \mathbf{b}_j) = (y_j, r_j)$ for each $j \in \set{1,\ldots,l}$,
representing one possible sequence of transitions from $s_0$ to $s_{i_0}$.

Then, for any $n\in\N$, when the sampler is run with an input bit stream
whose prefix is
$\mathbf{b}_1\cdots\mathbf{b}_l (\mathbf{c}_1\cdots\mathbf{c}_k)^n$
and a target distribution sequence whose prefix is $q_1 \cdots q_l (p_1 \cdots p_k)^n$,
it produces output symbols $(y_1 \cdots y_l) (x_1 \cdots x_k)^n$.
Define
\begin{equation}
N \defeq 1 + \left\lfloor \frac{\sum_{j=1}^l \abs*{\mathbf{b}_j}}
                               {\sum_{j=1}^k (\log(1/p_j(x_j)) - \abs{\mathbf{c}_j})} \right\rfloor
> \frac{\sum_{j=1}^l \abs{\mathbf{b}_j}}
       {\sum_{j=1}^k (\log(1/p_j(x_j)) - \abs{\mathbf{c}_j})},
\end{equation}
which is positive by \cref{eq:epsilon,eq:bad-cycle}.

When running the sampler on the target distribution sequence
$q_1 \cdots q_l (p_1 \cdots p_k)^N$ using uniform random bits,
the probability of generating the sequence $(y_1 \cdots y_l) (x_1 \cdots x_k)^N$
is at least
\begin{align}
2^{-\abs*{\mathbf{b}_1\cdots\mathbf{b}_l (\mathbf{c}_1\cdots\mathbf{c}_k)^N}}
&= 2^{-\sum_{j=1}^l \abs*{\mathbf{b}_j} - N \sum_{j=1}^k \abs{\mathbf{c}_j}} \\
&> 2^{-N \sum_{j=1}^k \log(1/p_j(x_j))} \\
&= (p_1(x_1)\cdots p_k(x_k))^N \\
&\geq (q_1(y_1)\cdots q_l(y_l)) (p_1(x_1)\cdots p_k(x_k))^N.
\end{align}
But generating the output $(y_1 \cdots y_l) (x_1 \cdots x_k)^N$
from the distributions $q_1 \cdots q_l (p_1 \cdots p_k)^N$
with probability greater than $(q_1(y_1)\cdots q_l(y_l)) (p_1(x_1)\cdots p_k(x_k))^N$
violates the exactness condition~\cref{eq:online-random-sampling-exact}.
\end{proof}

One last property we will need about the state information content is that it is nonnegative.
This result also verifies that every $A_{i,j} > -\infty$.
The proof is very similar to that of \cref{theorem:no-negative-cycles}.

\begin{lemma}[Nonnegativity of state information content]
\label{lemma:state-information-content-nonnegative}
Let $G$ be the directed graph with vertex set $\AState$
and edge weights $A_{i, j}$ between each pair of vertices $s_i, s_j$.
For each $s_i \in \AState$, the shortest paths from $s_0$ to $s_i$ have nonnegative weight,
i.e., $\delta(s_0, s_i) \geq 0$.
\end{lemma}

\begin{proof}
Assume, for the sake of contradiction, that there is a state $s \in \AState$
such that $\delta(s_0, s) < 0$.
Let $(s_0, s_{i_1}, s_{i_2}, \ldots, s_{i_k} = s)$ be a shortest path from $s_0$ to $s$ in $G$,
so that
\begin{align}
A_{0, i_1} + A_{i_1, i_2} + \cdots + A_{i_{k-1}, i_k} = \delta(s_0, s) < 0.
\end{align}
Define $\epsilon \defeq - (A_{0, i_1} + A_{i_1, i_2} + \cdots + A_{i_{k-1}, i_k}) / k$.
Then, for each $j \in \set{1, 2, \ldots, k}$, there exist $p_j \in \DAlphabet$, $x_j \in \N$,
and $\mathbf{c}_j \in \set{0,1}^*$ such that
$f(p_j, s_{i_{j-1}}, \mathbf{c}_j) = (x_j, s_{i_j})$ and
$\abs{\mathbf{c}_j} - \log(1/p_j(x_j)) < A_{i_{j-1}, i_j} + \epsilon$,
by the definition of $A_{i, j}$.
Then, when running the sampler on the target distribution sequence
$p_1 \cdots p_k$ using uniform random bits,
the probability of generating the sequence $x_1 \cdots x_k$
is at least
\begin{align}
2^{-\abs*{\mathbf{c}_1\cdots\mathbf{c}_k}} > 2^{-\sum_{j=1}^k \log(1/p_j(x_j))}= p_1(x_1)\cdots p_k(x_k),
\end{align}
by our choice of $\epsilon$, $p_j$, $x_j$, and $\mathbf{c}_j$.
But generating the output $x_1 \cdots x_k$
from the distributions $p_1 \cdots p_k$
with probability greater than $p_1(x_1)\cdots p_k(x_k)$
violates the exactness condition~\cref{eq:online-random-sampling-exact}.
\end{proof}

\Cref{theorem:difference-constraints-feasibility,
      theorem:no-negative-cycles,
        lemma:state-information-content-nonnegative}
imply the following result.

\begin{corollary}
\label{corollary:state-information-content-exists}
For any online random sampler, there exists a state information content function
as defined in \cref{def:state-information-content}.
\end{corollary}

\section{State Space Cycles and Information Loss}
\label{sec:cycles}

The state information content function lets us prove lower bounds on the
state space required to achieve a given entropy loss bound $\eps$.
In particular, generating any output sample with probability less than $1$
requires transitioning from one state to another
and/or consuming at least one fresh random bit from the input source.
For low-entropy target distributions, each transition must therefore lose information,
unless the fractional parts of the possible state information content values are sufficiently
dense in $\R/\Z$.
Two lemmas and a key theorem formalize these ideas.

\begin{lemma}[Information loss over cycles]
\label{lemma:cycle-loss}
Let $(s_{i_0}, s_{i_1}, \ldots, s_{i_k}=s_{i_0})$ be any cycle of states in $\AState$,
and let distributions $p_1, \ldots, p_k$, outputs $x_1, \ldots, x_k$, and
input bit sequences $\mathbf{c}_1, \ldots, \mathbf{c}_k$
satisfy $f(p_j, s_{i_{j-1}}, \mathbf{c}_j) = (x_j, s_{i_j})$ for each $j \in \set{1, 2, \ldots, k}$.
Then, the information lost over the course of iterating through the cycle once
can be lower bounded as follows:
\begin{equation}
\sum_{j=1}^k (\abs{\mathbf{c}_j} - \log(1/p_j(x_j))) \geq
\modZ*{- \sum_{j=1}^k \log(1/p_j(x_j))}.
\qedhere
\end{equation}
\end{lemma}

\begin{proof}
The string lengths $\abs{\mathbf{c}_j}$ are integers
and the sum on the left-hand side is nonnegative by \cref{theorem:no-negative-cycles}.
The result follows from the fact that if $x - y \ge 0$ where $x$ is an integer,
then $x \ge \ceil{y}$, so that
$x - y \ge \ceil{y} - y = -\floor{-y} - y \eqdef \modZ{-y}$.
\end{proof}

\begin{lemma}[State space information loss]
\label{lemma:state-space-loss}
When using an online random sampler to generate repeated samples
from a target distribution $p$,
the information lost over $\abs{\AState}$ samples is at least the
minimum loss incurred from iterating through a cycle of states in $\AState$:
\begin{equation}
\sum_{j=1}^{\abs{\AState}} (\abs{\mathbf{c}_j} - \log(1/p(x_j))) \geq
  \begin{aligned}[t]
    &h(s_{i_{\abs{\AState}}}) - h(s_{i_0}) \\
    &+ \inf\set*{
      \modZ*{\,-\sum_{\xinsupp} a_x \log(1/p(x))}
      \; \middle| \;
      a_x \in \N,
      1 \leq \sum_{\xinsupp} a_x \leq \abs{\AState}
      }
  \end{aligned}
\end{equation}
for every sequence of states $s_{i_0}, s_{i_1}, \ldots, s_{i_{\abs{\AState}}}$,
outputs $x_1, \ldots, x_{\abs{\AState}}$,
and input bit sequences $\mathbf{c}_1, \ldots, \mathbf{c}_{\abs{\AState}}$,
such that $f(p, s_{i_{j-1}}, \mathbf{c}_j) = (x_j, s_{i_j})$
for each $j \in \set{1, 2, \ldots, \abs{\AState}}$.
Here, $h$ is a state information content function for $f$.
\end{lemma}

\begin{proof}
By the pigeonhole principle, at least one state must be repeated among
$s_{i_0}, s_{i_1}, \ldots, s_{i_{\abs{\AState}}}$,
so we may take $L < R$ satisfying $s_{i_L} = s_{i_R}$.
Then, we can decompose the sequence of states into
an initial path from $s_{i_0}$ to $s_{i_L}$,
followed by a cycle from $s_{i_L}$ to itself (at index $i_R$),
followed by a final path from $s_{i_R}$ to $s_{i_{\abs{\AState}}}$.

Over the initial and final paths, the information difference can be bounded
by the state information content function $h$:
\begin{equation}
\sum_{j=1}^L (\abs{\mathbf{c}_j} - \log(1/p(x_j)))
\geq \sum_{j=1}^L (h(s_{i_j}) - h(s_{i_{j-1}})) = h(s_{i_L}) - h(s_{i_0}),
\end{equation}
and similarly
\begin{equation}
\sum_{j=R+1}^{\abs{\AState}} (\abs{\mathbf{c}_j} - \log(1/p(x_j)))
\geq h(s_{i_{\abs{\AState}}}) - h(s_{i_R}).
\end{equation}
The information difference over the cycle is bounded by \cref{lemma:cycle-loss}:
\begin{align}
\sum_{j=L+1}^R (\abs{\mathbf{c}_j} - \log(1/p(x_j)))
  &\geq \modZ*{- \sum_{j=L+1}^R \log(1/p(x_j))} \\
&\geq
   \inf\set*{
      \modZ*{\,-\sum_{\xinsupp} a_x \log(1/p(x))}
      \; \middle| \;
      a_x \in \N,
      1 \leq \sum_{\xinsupp} a_x \leq \abs{\AState}},
\end{align}
where the second inequality follows because the count vector
$a_x = \abs{\set{j \in \set{L+1, \ldots, R} \mid x_j = x}}$
for each $x \in \supp(p)$
is feasible for the infimum.
Combining these three bounds yields the result:
\begin{align}
~      & \sum_{j=1}^{\abs{\AState}} (\abs{\mathbf{c}_j} - \log(1/p(x_j))) \\
={}    & \sum_{j=1}^L (\abs{\mathbf{c}_j} - \log(1/p(x_j))) + \sum_{j=L+1}^R (\abs{\mathbf{c}_j} - \log(1/p(x_j))) + \sum_{j=R+1}^{\abs{\AState}} (\abs{\mathbf{c}_j} - \log(1/p(x_j))) \\
\geq{} & h(s_{i_L}) - h(s_{i_0}) + \modZ*{- \sum_{j=L+1}^R \log(1/p(x_j))} + h(s_{i_{\abs{\AState}}}) - h(s_{i_R}) \\
\geq{} & h(s_{i_{\abs{\AState}}}) - h(s_{i_0})
   + \inf\set*{
      \modZ*{\,-\sum_{\xinsupp} a_x \log(1/p(x))}
      \; \middle| \;
      a_x \in \N,
      1 \leq \sum_{\xinsupp} a_x \leq \abs{\AState}}.
\end{align}
\end{proof}

\begin{theorem}
\label{theorem:space-bound-gamma}
For a discrete probability distribution $p$ and an integer $M \ge 1$, define
\begin{equation}
\Gamma_p(M) \defeq
            \inf\set*{
              \modZ*{\,-\sum_{\xinsupp} a_x \log(1/p(x))}
              \;\middle|\;
              a_x \in \N,
              1 \leq \sum_{\xinsupp} a_x \leq M}.
\label{eq:gamma-p}
\end{equation}
Let $f$ be an online random sampler with state space $\AState$,
and $p \in \DAlphabet$.
If $f$ achieves entropy loss at most $\eps$ bits per sample, then
\begin{equation}
\abs{\AState} \ge \frac{\Gamma_p(\abs{\AState})}{\eps}.
\qedhere
\end{equation}
\end{theorem}

\begin{proof}
Run the sampler to generate $n\abs{\AState}$ i.i.d.~samples from $p$.
Partition these samples into $n$ consecutive blocks of $\abs{\AState}$ samples each.
For each block $b \in \set{0,\dots,n-1}$ starting at state $s_{b\abs{\AState}}$ and ending at
$s_{(b+1)\abs{\AState}}$, \cref{lemma:state-space-loss} gives
\begin{equation}
\sum_{j=b\abs{\AState}+1}^{(b+1)\abs{\AState}}\left(\abs{\mathbf{c}_j} - \log(1/p(x_j))\right)
   \ge h(s_{(b+1)\abs{\AState}}) - h(s_{b\abs{\AState}}) + \Gamma_p(\abs{\AState}).
\end{equation}
Summing over the blocks $b=0,\dots,n-1$, the $h$ terms telescope as
\begin{equation}
\sum_{j=1}^{n \abs{\AState}}\left(\abs{\mathbf{c}_j} - \log(1/p(x_j))\right)
   \ge h(s_{n\abs{\AState}}) - h(s_0) + n\Gamma_p(\abs{\AState}) \ge n \Gamma_p(\abs{\AState}),
\end{equation}
where the second inequality follows from the fact that $h(s_0) = 0$ and $h(s_{n\abs{\AState}}) \ge 0$.
Taking expectations gives entropy loss at least $n \Gamma_p(\abs{\AState})$
for generating $n\abs{\AState}$ samples.
By the assumption on the entropy loss we have $n \Gamma_p(\abs{\AState}) \le \eps n\abs{\AState} + o(n\abs{\AState})$.
Dividing by $\eps n$ and taking $n \to \infty$ gives $\Gamma_p(\abs{\AState})/\eps \le \abs{\AState}.$
\end{proof}

\section{Lower Bound for Bernoulli Outputs}
\label{sec:bernoulli}

We now apply \cref{theorem:space-bound-gamma} to prove
\cref{theorem:space-lb}.

\begin{nicebox}
\getkeytheorem{theorem:space-lb}
\end{nicebox}

\begin{proof}
To establish \zcref[noname]{item:state-space-lower-bounds-1},
consider the target distribution $p \defeq \mathrm{Bernoulli}(1/3)$.
The fractional part of the negative information content of each outcome is
$\modZ{-\log(3)} = \modZ{-\log(3/2)} = 2-\log(3)$.
A useful result from number theory, discussed further in \cref{sec:bernoulli-factors},
is that
\begin{align}
\min_{1 \leq a \leq M} \modZ*{-a \log 3}
\geq M^{-4.116201 - o(1)},
\label{eq:apply-diophantine-approx}
\end{align}
where the $o(1)$ term goes to $0$ as $M \to \infty$.
Therefore, \cref{theorem:space-bound-gamma} gives the result:
\begin{align}
\abs{\AState}
   \ge \frac{\Gamma_p(\abs{\AState})}{\eps}
   \ge \frac{\abs{\AState}^{-4.116201 - o(1)}}{\eps}
\quad\implies\quad
\log(\abs{\AState})
   \ge \left(\frac{1}{5.116201} - o(1)\right) \log\frac{1}{\eps},
\end{align}
where the $o(1)$ term now goes to $0$ as $\eps \to 0$.

To establish \zcref[noname]{item:state-space-lower-bounds-2},
consider the target distribution
$p_m \defeq \mathrm{Bernoulli}(1/(2^m+1))$ for some $m \in \N$.
The fractional part of the negative information content of each outcome is
$\modZ*{-\log(2^m+1)} = \modZ*{-\log(1+2^{-m})} = 1 - \log(1+2^{-m}) \eqdef 1 - \alpha_m$.
For any $M \geq 1$, we have
\begin{equation}
\Gamma_{p_m}(M)
= \min_{1 \le a \le M}\modZ*{-a\thinspace \alpha_m}
\geq \min_{1 \le a \le M}(1 - a\thinspace \alpha_m)
= 1 - M\alpha_m.
\end{equation}
Let $\mathcal{M} \subset \N$ denote the set of all $m$
for which $p_m \in \DAlphabet$.
There exist infinitely many such $m$ by assumption,
so taking $m \to \infty$ along the subsequence ensures that $\alpha_m \to 0$.
Therefore, \cref{theorem:space-bound-gamma} gives the reuslt:
\begin{align}
\abs{\AState}
   \ge \frac{\Gamma_{p_m}(\abs{\AState})}{\eps}
   \ge \frac{1 - \abs{\AState}\thinspace \alpha_m}{\eps}
\quad\implies\quad
\abs{\AState}
   \ge \sup_{m \in \mathcal{M}} \frac{1 - \abs{\AState}\thinspace \alpha_m}{\eps}
   = \frac{1}{\eps}.
\end{align}
\end{proof}

\subsection{Constant Factors and Diophantine Approximation}
\label{sec:bernoulli-factors}

The bound in \zcref[noname]{item:state-space-lower-bounds-1} of \cref{theorem:space-lb}
makes use of the fact \cref{eq:apply-diophantine-approx} that integer multiples of $\log(3)$ cannot be ``too close'' to integers,
or equivalently, that $\log(3)$ cannot be approximated too well
by rational numbers with small denominators.
More precisely, \citeauthor{baker1975}'s theorem on linear forms in logarithms \citep[Theorem 3.1]{baker1975}
shows that there exists a constant $c > 0$ such that
\begin{align}
\modZ{-a\log(3)} \geq \frac{1}{a^c}
&&
(a = 1, 2, \dots).
\end{align}
Approximation by rational numbers, also called Diophantine approximation,
is an important topic studied in number theory.
The hardness of approximation of a single number (resp.~a set of numbers)
is measured by the irrationality exponent (resp.~the linear independence measure),
which we define following the exposition of \citet{wu2014}.

\begin{definition}
\label{def:irrationality-exponent}
The \textit{irrationality exponent} of an irrational real number $\alpha$, denoted $\mu(\alpha)$,
is the minimum real number such that, for every $\epsilon > 0$,
there exists $q_0(\epsilon)$ such that
\begin{equation}
\abs*{\alpha - \frac{p}{q}} \geq \frac{1}{q^{\mu(\alpha) + \epsilon}}
\end{equation}
for all integers $p$ and $q$ with $q > q_0(\epsilon)$.
\end{definition}

\begin{definition}
\label{def:linear-independence-measure}
The \textit{linear independence measure} of a set of $\Q$-linearly independent
real numbers $\alpha_0, \ldots, \alpha_k$, denoted $\nu(\alpha_0, \ldots, \alpha_k)$,
is the minimum real number such that, for every $\epsilon > 0$,
there exists $H_0(\epsilon)$ such that
\begin{equation}
\abs*{p\alpha_0 + q_1\alpha_1 + \cdots + q_k\alpha_k} \geq H^{-\nu(\alpha_0, \ldots, \alpha_k) - \epsilon}
\end{equation}
for all integers $p, q_1, \ldots, q_k$ satisfying $H > H_0(\epsilon)$,
where $H \defeq \max\set{\abs{q_1},\ldots,\abs{q_k}}$.
\end{definition}

\begin{remark}
\label{remark:modZ-bound-from-linear-independence-measure}
An immediate consequence of the definition of a linear independence measure is that
\begin{equation}
\modZ*{q_1\alpha_1 + \cdots + q_k\alpha_k}
\geq \max\set{\abs{q_1},\ldots,\abs{q_k}}^{-\nu(1, \alpha_1, \ldots, \alpha_k) - o(1)}
\end{equation}
where $o(1)$ describes the limiting behavior as
$\max\set{\abs{q_1},\ldots,\abs{q_k}} \to \infty$.
\end{remark}

Although \citet{baker1975} provided the first general proof that
sets of logarithms of algebraic numbers have finite linear independence measures,
the constants were far from optimal.
Extensive work has been devoted to improving the constant bounds for specific numbers,
and the current best known bound on the irrationality exponent of $\log(3)$ is
$\mu(\log(3)) \leq 5.116201$, which is a corollary of the stronger result
$\nu(1, \ln(2), \ln(3)) \leq 4.116201$ \citep[Theorem 1]{bondareva2018}.
This bound implies \cref{eq:apply-diophantine-approx},
which yields the space lower bound of $(1/5.116201 - o(1))\log(1/\eps)$ bits
when sampling from $\mathrm{Bernoulli}(1/3)$.

\section{Lower Bound for i.i.d.~Outputs}
\label{sec:iid}

We now generalize the Diophantine approximation analysis
to prove the space lower bound in \cref{theorem:space-lb-iid},
which applies to a large class of samplers that emit i.i.d.~outputs
from a fixed distribution.

When considering i.i.d.~samples from a target distribution $p$
where $\set{\modZ{\log p(x)} \mid x \in \supp(p)}$ contains only one element,
it suffices to analyze the irrationality exponent of
the single number $\modZ{\log p(x)}$
to obtain a state-space lower bound.
In general, when $\abs{\set{\modZ{\log p(x)} \mid x \in \supp(p)}} > 1$,
our proof method from \cref{theorem:space-lb}
requires analyzing the linear independence measure
of the set $\set{1} \cup \set{\modZ{\log p(x)} \mid x \in \supp(p)}$.
With a more careful analysis, we can attain lower bounds even if there is
a certain $\Q$-linear dependence among the $\modZ{\log p(x)}$ values,
as long as their $\Q$-linear combinations with nonnegative coefficients
do not include any nonzero integers, as shown in the following theorem.

\begin{theorem}[Space lower bound from linear independence measure]
\label{theorem:space-linear-independence}
Let $p$ be a finite discrete distribution over support $\set{x_1, \ldots, x_k}$.
Let $\set{1, \alpha_1, \ldots, \alpha_m}$ be a $\Q$-linearly independent
set of real numbers and $q \in \N_{>0}$, such that
for each $i \in \set{1,\ldots,k}$ we can write
\begin{equation}
\log p(x_i)
= \beta_i / q + \lambda_{i,1} \alpha_1 + \ldots + \lambda_{i,m} \alpha_m
\end{equation}
for some coefficients $\beta_i \in \Z$ and $\lambda_{i,j} \in \N$ with $\sum_{j=1}^m \lambda_{i,j} > 0$.
If $\set{1, \alpha_1, \ldots, \alpha_m}$ has finite linear independence measure
$C \defeq \nu(1, \alpha_1, \ldots, \alpha_m)$,
then generating i.i.d.~samples from $p$ with an entropy loss bound of $\eps$ bits per sample
requires a state space of size
$\abs{\AState} \geq (1/\eps)^{1/(1+C)-o(1)}$.
\end{theorem}

\begin{proof}
\allowdisplaybreaks

By \cref{theorem:space-bound-gamma}, it suffices to show there
exists a function $\delta(M) \to 0$ as $M\to\infty$ such that
\begin{equation}
\Gamma_p(M) \ge M^{-C-\delta(M)}.
\label{eq:gamma-p-bound}
\end{equation}
After setting $M = \abs{\AState}$, \cref{theorem:space-bound-gamma} gives
\begin{align}
\abs{\AState} \ge \frac{\Gamma_p(\abs{\AState})}{\eps} \ge \frac{\abs{\AState}^{-C-\delta(\abs{\AState})}}{\eps}
\quad\implies\quad
\abs{\AState} \ge \left(\frac{1}{\eps}\right)^{1/(C+1+\delta(\abs{\AState}))} \geq \left(\frac{1}{\eps}\right)^{1/(C+1) - o(1)}.
\end{align}
The linear-independence hypotheses imply that $\Gamma_p(M)>0$ for every fixed
$M$ so \cref{theorem:space-bound-gamma} rules out bounded $\abs{\AState}$
as $\eps\to0$.
Therefore $\abs{\AState}\to\infty$ as $\eps\to0$, so
$\delta(\abs{\AState})=o(1)$ in the $\eps\to0$ limit.

We now prove \cref{eq:gamma-p-bound}:
\begin{align}
\Gamma_p(M)
={}    & \min\set*{\modZ*{\sum_{i=1}^k a_i \log p(x_i)}
          \; \middle| \;
          a_i \in \N, \,
          1 \leq \sum_{i=1}^k a_i \leq M} \\
={}    & \min\set*{\modZ*{\sum_{i=1}^k a_i \left(\frac{\beta_i}{q} + \sum_{j=1}^m \lambda_{i,j}\alpha_j \right)}
          \; \middle| \;
          a_i \in \N, \,
          1 \leq \sum_{i=1}^k a_i \leq M} \\
\geq{} & \frac{1}{q} \min\set*{\modZ*{\sum_{i=1}^k a_i q\sum_{j=1}^m \lambda_{i,j}\alpha_j}
          \; \middle| \;
          a_i \in \N, \,
          1 \leq \sum_{i=1}^k a_i \leq M} \label{eq:extract-q} \\
={}    & \frac{1}{q} \min\set*{\modZ*{ \sum_{j=1}^m \left(q\sum_{i=1}^k a_i \lambda_{i,j}\right) \alpha_j}
          \; \middle| \;
          a_i \in \N, \,
          1 \leq \sum_{i=1}^k a_i \leq M} \\
\geq{} & \frac{1}{q} \min\set*{\modZ*{\sum_{j=1}^m b_j \alpha_j}
          \; \middle| \;
          b_j \in \N, \,
          1 \leq \sum_{j=1}^m b_j \leq M q \sum_{i=1}^k \sum_{j=1}^m \lambda_{i,j}} \label{eq:coalesce-coefficients} \\
\geq{} & M^{-C - \delta(M)}. \label{eq:apply-linear-independence-measure}
\end{align}
We can extract the factor of $1/q$ in \cref{eq:extract-q} using the identity
$\modZ{z/q+r} \geq \modZ{z+rq}/q = \modZ{rq}/q$ for $r \in \R$ and $q \in \N_{>0}$ and $z \in \Z$.
We then reduce to \cref{eq:coalesce-coefficients} by reparameterizing
and bounding the new coefficients $b_j$.
Lastly, \cref{eq:apply-linear-independence-measure} follows from
\cref{remark:modZ-bound-from-linear-independence-measure}. The fixed
constants $1/q$ and $q\sum_{i=1}^k\sum_{j=1}^m \lambda_{i,j}$ are absorbed
by increasing $\delta(M)$ by $O(1/\log M)$, which still leaves
$\delta(M) \to 0$.
\end{proof}

The Khintchine-Groshev theorem \citep{groshev1938} shows that
the linear independence measure of almost all tuples of $k$ real numbers (together with $1$) is $k$.
This result nearly suffices to lower-bound the space complexity for i.i.d.~sampling
from almost all finite discrete distributions, except that
the set of probability distributions over $k$ outcomes is a $(k-1)$-dimensional
manifold, which has measure zero in $\R^k$.
Therefore, showing that $\nu(1, \modZ{\log p(x_1)}, \ldots, \modZ{\log p(x_k)}) = k$
for almost all distributions requires a stronger result providing lower bounds
for Diophantine approximation on manifolds.

\begin{theorem}[Diophantine approximation on manifolds \citep[Theorem A]{kleinbock1998}]
\label{theorem:diophantine-approximation-on-manifolds}
Let $U \subseteq \R^d$ be an open set,
and let $\mathbf{f}(U) = \set{\mathbf{f}(\mathbf{u}) \mid \mathbf{u} \in U}$
be a submanifold of $\R^k$.
Assume that for almost all $\mathbf{u} \in U$,
the set of partial derivatives of $\mathbf{f}$ at $\mathbf{u}$
up to some fixed finite order spans $\R^k$.
Then, for almost all $\mathbf{u} \in U$, the numbers
$(1, f_1(\mathbf{u}), \ldots, f_k(\mathbf{u}))$ are $\Q$-linearly independent
with linear independence measure
\begin{equation}
\nu(1, f_1(\mathbf{u}), \ldots, f_k(\mathbf{u})) \leq k,
\end{equation}
where $(f_1(\mathbf{u}), \ldots, f_k(\mathbf{u})) = \mathbf{f}(\mathbf{u})$.
\end{theorem}

We can now prove \cref{theorem:space-lb-iid}, which is restated below.

\begin{nicebox}
\getkeytheorem{theorem:space-lb-iid}
\end{nicebox}

\begin{proof}
Apply \cref{theorem:diophantine-approximation-on-manifolds} with
\begin{align}
U &= \set*{(p_1,\ldots,p_{k-1}) \mid p_i > 0, \sum_{i=1}^{k-1} p_i < 1}; \\
\mathbf{f}(p_1,\ldots,p_{k-1}) &= \left(\log(p_1),\ldots,\log(p_{k-1}), \log\!\left(1-\sum_{i=1}^{k-1} p_i\right)\right).
\end{align}
At any point, the first-order partial derivatives are
\begin{equation}
\frac{\partial}{\partial p_j} \mathbf{f}(p_1,\ldots,p_{k-1}) = \frac{1}{\ln(2)} \left(\ldots,p_j^{-1},\ldots,-\left(1-\sum_{i=1}^{k-1} p_i\right)^{-1}\right),
\end{equation}
where only the $j$th and $k$th coordinates are nonzero.
Furthermore, the second-order partial derivative with respect to $p_1$ is
\begin{equation}
\frac{\partial^2}{\partial p_1^2} \mathbf{f}(p_1,\ldots,p_{k-1}) = \frac{1}{\ln(2)} \left(-p_1^{-2},0,\ldots,0,-\left(1-\sum_{i=1}^{k-1} p_i\right)^{-2}\right),
\end{equation}
which, together with the $k-1$ first-order partial derivatives,
forms a linearly independent set and therefore spans $\R^k$.
Therefore, the theorem applies to this manifold and yields
\begin{equation}
\nu(1,\log(p_1),\ldots,\log(p_k)) \leq k
\end{equation}
for almost all finite discrete distributions $(p_1,\ldots,p_k)$.

Finally, \cref{theorem:space-linear-independence} with
$(\alpha_1,\ldots,\alpha_m)=(\log(p_1),\ldots,\log(p_k))$ yields
\begin{equation}
\abs{\AState} \geq (1/\eps)^{1/(1+k)-o(1)}
\end{equation}
for almost all distributions, and taking the logarithm gives the desired result.
\end{proof}

\section{Related Work}
\label{sec:related-work}

\paragraph{Entropy-Efficient Sampling}

Previous work on random sampling has mainly focused on either entropy efficiency
or computational complexity, with algorithms matching lower bounds on both fronts.
\Citet{knuth1976} derive a DDG tree algorithm that uses the exact minimum
possible number of input random bits, whether sampling from a single
distribution or from a sequence of distributions;
however, this algorithm uses unbounded space and per-sample computation time
when sampling an infinite sequence \citep[Appendix A]{draper2026soda}.
Computationally optimal sampling algorithms depend on
the computational model and format of the target distribution;
of particular note is the \citeauthor{walker1977} alias method \citep{walker1977},
which represents the target distribution in a linear-size data structure and
generates each sample using a constant number of word-level operations,
matching a trivial lower bound.
Although the original alias method has significant entropy loss,
the randomness recycling method of \citet[Algorithm 5.2]{draper2026soda}
reduces the entropy loss to $\eps$ bits per sample while
increasing the space by $O(\log(1/\eps))$ bits and adding
a constant number of $O(\log(1/\eps))$-bit arithmetic operations per sample.
Therefore, our lower bound in \cref{theorem:space-lb}
shows that for rational distributions,
the space usage of randomness recycling~\citep{draper2026soda}
is optimal up to constant factors,
and the runtime is optimal, up to the cost of arithmetic on the state.

\paragraph{Expected Space and Irrational Probabilities}

Our lower bounds refer to the \textit{persistent} space $s \in \AState$
that is carried by a sampler across each round of online random sampling.
Sampling algorithms may also use \textit{temporary} space, which is
allocated and released within a single round.
While temporary space must be unbounded to handle irrational probabilities,
its expectation may be finite.
For this setting, \citet{dgh2026} study the problem of online sampling
given CDF oracles that produce dyadic rational approximations to
irrational probabilities, up to any precision.
They show that over $n$ rounds of online random sampling, the expected overall
space (i.e., maximum of temporary and persistent) must grow as
$\Omega(\log n)$~\citep[Theorem 7.3]{dgh2026}.
They also introduce an algorithm that matches both their lower
bound on expected space and our lower bound on persistent space, using
deterministically $\Theta(\log(1/\eps))$ bits of persistent space to
achieve an entropy loss bound of $\eps$ bits per sample.
Their algorithm \citep[Algorithm 2]{dgh2026} uses an iteration counter $t$
to scale a parameter $\Delta_t$ which upper-bounds $M$,
in addition to the uniform state $(Z,M)$, and thus its persistent space
is $(3+o(1))\log(1/\eps)$ bits.
However, a variant with constant $\Delta = \widetilde{\Theta}(1/\epsilon)$
requires only $(2+o(1))\log(1/\eps)$ bits of persistent space $(Z,M)$,
which is optimal up to a factor of $2$ by
\cref{theorem:space-lb}, \zcref[noname]{item:state-space-lower-bounds-2}.

\paragraph{Single-Variable Sampling}

Orthogonally to this article's focus on asymptotic per-sample entropy loss bounds,
several works have studied the non-asymptotic efficiency of
generating a single sample from a distribution.
For this case,
the \citeauthor{knuth1976} DDG tree algorithm has an entropy loss bounded by $2$.
Explicitly representing the entire DDG tree
as in \citet{roy2013} requires exponential space~\citep[Theorem 3.5]{saad2020popl},
if the target distribution has finite support and arbitrary rational probabilities.
This space cost can be bypassed by instead lazily generating the tree during traversal
\Parencites[Algorithm 1]{saad2025}[p.~384]{knuth1976} (at the cost of higher runtime
during sampling) or using rejection sampling~\citep{saad2020fldr,draper2026ieee}
(at the cost of slightly higher entropy cost).

\paragraph{Lower Bounds in Coding Theory}
There exist results in the information theory literature
that provide lower bounds on the arithmetic precision or delay
(similar to our space complexity) of entropy-efficient codes.
\Citet{reznik2007} shows that for arithmetic codes,
achieving a redundancy (similar to our entropy loss per sample) of $\eps$ requires
$\Omega(\log(1/\eps))$ bits of precision for representing the probabilities;
although their model is quite restricted,
similar techniques may show that online random samplers based on
discretizations of the interval method of \citet{han1997}
require $\Omega(\log(1/\eps))$ space.
\Citet{shayevitz2014} show that for almost all distributions,
achieving a redundancy of $\eps$ requires $\Omega(\log(1/\eps))$ delay;
although their lower bound is quite general,
and a delay of $\Omega(\log(1/\eps))$
would imply a space complexity of $\Omega(\log(1/\eps))$
if each symbol had to be explicitly stored over the entire window of the delay,
an efficient coder may compress the symbols at intermediate stages
to achieve a lower space complexity.
Random sampling and coding theory are closely related but subtly distinct,
and to our knowledge the results in our paper do not imply, nor are implied
by, the aforementioned results in coding theory.
It remains an interesting question to what extent lower bounds or proof
techniques can be adapted from one setting to the other.

\section{Open Questions}
\label{sec:open-questions}

\paragraph{Distribution Dependence}
It is well known that zero entropy loss can be achieved for any dyadic distribution
(for which each probability is in $\set{0}\cup\set{2^{-n} \mid n \in \N}$)
using no persistent state \citep{knuth1976}.
Our analysis of entropy loss purely in terms of the state information
content function does not rule out the possibility of a sampler that
achieves an entropy loss of $\eps = 0$ for particular non-dyadic
distributions using a finite state space.
However, we conjecture that such a sampler does not exist.

\paragraph{Constant Factors}
Another open question is what constant factor for the $\Omega(\log(1/\eps))$
space is necessary, and to what extent it depends on the target distribution set $\DAlphabet$.
\Citet{draper2026soda} show that $2$ is an upper bound on the constant factor
whenever $\DAlphabet$ consists of finitely many rational distributions over finite support,
and a variant on the method of \citet{dgh2026} achieves this same bound for arbitrary distributions.
This constant factor could potentially be reduced to $1$,
according to \cref{theorem:space-lb}, \zcref[noname]{item:state-space-lower-bounds-2}.
On the lower-bound side,
the state information content cannot take entirely arbitrary values,
because there must be some relationship between states
that allows the algorithm to generate correct samples.
For example, a state space of size $\abs{\AState}=2$
is no better than one of size $\abs{\AState}=1$,
as a consequence of the exactness condition in \cref{def:online-random-sampling}.
A more fine-grained analysis may therefore strengthen the
lower bounds we have presented.

\paragraph{Continuous Distributions}
An orthogonal research direction is sampling from continuous distributions,
iteratively outputting new bits in the binary expansion of the generated real number.
In particular, given a cumulative distribution function $F$ satisfying $F(0)=0$
and $F(1)=1$, the problem is to generate outputs $X_1,X_2,\ldots$ following the conditional distributions
\begin{equation}
(X_{n+1} \mid X_1,\ldots,X_n) \sim \mathrm{Bernoulli}\left(
\frac{F(2^{-n}+\sum_{i=1}^{n}X_i/2^i) - F(2^{-n-1}+\sum_{i=1}^{n}X_i/2^i)}
     {F(2^{-n}+\sum_{i=1}^{n}X_i/2^i) - F(\sum_{i=1}^{n}X_i/2^i)}
\right).
\end{equation}
\Citeauthor{knuth1976} study finite-state generators for sampling from
particular distributions, and pose several related questions about
which distributions can be sampled in this way and how large the state
space must be in order to come within $\eps$ of the entropy-optimal algorithm
\citep[p.~427]{knuth1976}.
Question `(v)', about which distributions have finite-state generators,
was resolved by \citet{kindler2004},
who showed that a smooth CDF has a finite-state generator
if and only if it is a polynomial with rational coefficients,
whose derivative has no irrational roots on $[0,1]$.
However, question `(iii)', about which distributions have entropy-optimal
finite-state generators; and question `(iv)', about how large of a state space is needed
in order to come within $\eps$ of the entropy-optimal algorithm for $F(x)=x^3$;
remain unanswered to the best of our knowledge.
Our lower bounds depend crucially on the fact that an online random sampler
must be able to produce correct samples from any of its target distributions
at any time step, whereas for continuous samplers the output bits
approach the $\mathrm{Bernoulli}(1/2)$ distribution, which allows
significant optimizations.
In fact, the $\eps$ mentioned by \citeauthor{knuth1976} is a uniform bound on
the entropy loss relative to the optimal algorithm at each iteration count,
rather than a per-iteration loss which accumulates as $\eps n$ in our problem setting.
Therefore, entropy-space lower bounds for continuous distributions remain
an open question.

\section*{Acknowledgments}
This material is based upon work supported by the National Science Foundation
under Award No.~2311983.
Any opinions, findings and conclusions or recommendations expressed in this
material are those of the authors and do not necessarily reflect the views
of the National Science Foundation.

\printbibliography

\appendix

\end{document}